\documentclass[a4paper,11pt]{article}

\usepackage{jheppub}

\makeatletter
\gdef\@fpheader{}
\makeatother

\usepackage{color}

\usepackage{comment}

\usepackage{tikz-cd}

\usepackage[bibtex-style]{amsrefs}
\usepackage{amsthm}

\newcommand{\doi}[1]{\href{https://doi.org/#1}{#1}}

\usepackage[shortlabels]{enumitem}
\setlist{noitemsep}

\newtheorem*{theorem}{Theorem}
\newtheorem*{lemma}{Lemma}
\newtheorem*{corollary}{Corollary}
\newtheorem*{definition*}{Definition}

\newcommand{\half}{\tfrac12}

\newcommand{\<}{\mathopen{\{\!\!\{}}
\renewcommand{\>}{\mathclose{\}\!\!\}}}
\renewcommand{\(}{\mathopen{(\!(}}
\renewcommand{\)}{\mathclose{)\!)}}
\renewcommand{\[}{\mathopen{[\![}}
\renewcommand{\]}{\mathclose{]\!]}}

\newcommand{\CD}{\mathcal{D}}

\newcommand{\CO}{\mathcal{O}}

\renewcommand{\S}{\mathsf{S}}
\newcommand{\D}{\mathsf{D}}

\newcommand{\R}{\mathbb{R}}
\newcommand{\g}{\mathfrak{g}}
\newcommand{\h}{\mathfrak{h}}

\newcommand{\C}{\mathbf{C}}
\newcommand{\M}{\mathbf{M}}

\newcommand{\Sym}{S}

\DeclareMathOperator{\End}{End}

\DeclareMathOperator{\so}{\mathbf{so}}

\title{The Moyal cohomology of the {\boldmath$N=1$} spinning particle}

\author{Ezra Getzler}

\affiliation{Department of Mathematics, Northwestern University,
  Evanston, IL 60208, USA}

\emailAdd{getzler@northwestern.edu}

\abstract{The Batalin--Fradkin--Vilkovisky formalism studies a
  differential graded symplectic supermanifold whose cohomology
  vanishes in negative degree. This hypothesis is violated by the
  $N=1$ spinning particle; its cohomology is nontrivial in all
  negative degrees \citep{curved}.  Replacing the Poisson bracket by
  the Moyal bracket has the effect of eliminating these cohomology
  classes.}


\notoc

\begin{document}

\maketitle

\clearpage

\section{Introduction}

\subsection{The {\boldmath$N=1$} spinning particle}

In this article, we study the $N=1$ spinning particle in a background
Lorentzian metric and magnetic field. We work in an open simply
connected subset $U$ of Minkowski space $\R^{D-1,1}$. Let $\eta_{ab}$ be
the Minkowski metric in the standard orthonormal basis
$\{v_a\mid 1\le a\le D\}$. The background Lorentzian metric is expressed by
a moving frame
\begin{equation*}
  e^a = e^a_\mu \, dx^\mu \in \Omega^1(U,\R^{D-1,1}) ,
\end{equation*}
and the background magnetic field by a one-form
$A = A_\nu \, dx^\mu \in \Omega^1(U)$. The Levi-Civita connection
\begin{equation*}
  \omega^a{}_b = \omega^a_{\mu b} \, dx^\mu \in \Omega^1(U,\so(D-1,1)) ,
\end{equation*}
is determined from $e$ by Cartan's first structure equation
\begin{equation*}
  de^a + \omega^a{}_b \wedge e^b = 0 ,
\end{equation*}
and the curvature is determined by Cartan's second structure equation
\begin{equation*}
  R^a{}_b = d\omega^a{}_b + \omega^a{}_c \wedge \omega^c{}_b \in \Omega^2(U,\so(D-1,1)) .
\end{equation*}
Set
\begin{align*}
  \omega_{ab} &= \eta_{ac} \, \omega^c{}_b = \omega_{\mu ab} \, dx^\mu , &
  R_{ab} &= \eta_{ac} \, R^c{}_b = \half R_{\mu\nu ab} \, dx^\mu \wedge dx^\nu .
\end{align*}
The magnetic field $B=dA$ is the differential of $A$:
\begin{equation*}
  B = \half B_{\mu\nu} dx^\mu \wedge dx^\nu \in \Omega^2(U) .
\end{equation*}

Let $M$ be the supermanifold $M=T^*U\times\Pi\R^{D-1,1}$, with coordinates
$\{x^\mu,p_a,\theta^a\}$.  The differential $\Omega=d\lambda$ of the one-form
\begin{equation*}
  \lambda = p_a e^a + A + \half \omega_{ab} \theta^a \theta^b - \half \eta_{ab} \theta^a d\theta^b \in
  \Omega^1(M)
\end{equation*}
is a symplectic form on $M$:
\begin{align*}
  \Omega &= dp_a \wedge e^a + p_a \, de^a - B + \half \omega_{ab} \, \theta^a d\theta^b + \half
      d\omega_{ab} \theta^a \theta^b - \half \eta_{ab} d\theta^a \wedge d\theta^b \\
    &= dp_a \wedge e^a - p_a \omega^a{}_b \wedge e^b + B + \half \omega_{ab} \, \theta^a d\theta^b +
      \half \bigl( R_{ab} - \omega_a{}^c \wedge \omega_{cb} \bigr) \theta^a \theta^b
      - \half \eta_{ab} d\theta^a \wedge d\theta^b .
\end{align*}
The associated Poisson bracket on $M$ equals
\begin{align*}
  & \begin{aligned}
    & \{\theta^a,\theta^b\} = \eta^{ab} , &
    & \{p_a,x^\mu\} = f_a^\mu , &
    & \{p_a,\theta^b\} = - f^\mu_a \omega_{\mu c}^b \theta^c ,
  \end{aligned} \\
  & \{p_a,p_b\} = \bigl( f^\mu_a \omega_{\mu b}^c - f^\mu_b \omega_{\mu a}^c \bigr)
    p_c - f^\mu_a f^\nu_b \bigl( \half R_{\mu\nu cd} \theta^c \theta^d + B_{\mu\nu} \bigr) ,
\end{align*}
(We use the negative of the classical Poisson bracket in this article,
ensuring that the Hamiltonian vector field $X_a$, characterized by
$\iota(X_a)\omega=-da$, satisfies $X_a(b)=\{a,b\}$.)

The odd parity function $q=p_a\theta^a$ has Poisson brackets
\begin{align*}
  \{q,x^\mu\} &= f_a^\mu \theta^a , \quad
  \{q,\theta^a\} = \delta^{ab} p_b - f^\mu_b \omega_{\mu c}^a \theta^b \theta^c , \\ \notag
  \{q,p_a\} &= \bigl( f^\mu_b \omega_{\mu a}^c - f^\mu_a \omega_{\mu b}^c \bigr) \theta^b
  p_c - f_b^\mu f_a^\nu B_{\mu\nu} \theta^b .
\end{align*}
In particular, its Poisson bracket with itself is given by the formula
\begin{equation*}
  h = \half \{q,q\} = \half \eta^{ab}p_ap_b - \half f^\mu_a f^\nu_b B_{\mu\nu} \theta^a\theta^b .
\end{equation*}

If $M$ is a symplectic supermanifold and $\g$ is a finite-dimensional
Lie superalgebra, a moment map $\mu:M\to\g^\vee$ is a Poisson map from
$M$ to the Poisson supermanifold $\g^\vee$. In this article, we only
consider moment maps for the Lie superalgebra $\h$ governing $N=1$
supersymmetric quantum mechanics: this is the free Lie superalgebra
generated by a single odd generator $Q$, spanned by $Q$ and
$H=\half[Q,Q]$. Such a moment map is determined by an odd function $q$
on $M$, which generates the vector field $Q=\{q,-\}$. The bracket
$h=\half\{q,q\}$ generates the vector field
\begin{align*}
  H &= \half[Q,Q] \\
    &= \{h,-\} = \{q,\{q,-\}\} = \half \{\{q,q\},-\} .
\end{align*}

\subsection{Differential graded Poisson superschemes}

The algebra of functions on a differential graded (dg) superscheme
$\M$ is a differential graded (dg) superalgebra $\CO$, filtered by
powers of the dg ideal $\CO_+=F^1\CO$ generated by elements of $\CO$
of strictly positive degree. We assume that this filtration is
complete:
\begin{equation*}
  \CO = F^0\CO \supset F^1\CO \supset \cdots
\end{equation*}

A dg Poisson superscheme $\M$ is a dg superscheme whose dg
superalgebra $\CO$ of functions is a dg Poisson superalgebra, and the
differential $\delta=\{S,-\}$ is a Hamiltonian vector field, with
Hamiltonian function $S\in\CO_+$. Since
\begin{equation*}
  \{F^i\CO,F^j\CO\}\subset F^{i+j-1}\CO ,
\end{equation*}
we see that $\CO_+=F^1\CO$ is a coisotropic ideal: it is closed under
Poisson bracket.

The vanishing locus of the coisotropic ideal $\CO_+$ is a coisotropic
dg superscheme $\M_-\subset\M$. The dg superalgebra of functions
$\CO/\CO_+$ of the dg superscheme $\M_-$ is concentrated in
nonpositive degree. If $\M_-$ is smooth, it is a dg supermanifold of
the type introduced by Tate \citep{Tate} in the ungraded case. These
are the basic objects of derived geometry. The dg superscheme $\M$ may
be thought of as an $L_\infty$-algebroid over $\M_-$, with
Chevalley-Eilenberg complex $\CO$.

Let $\CO_-\subset\CO$ be the ideal generated by elements of $\CO$ of
negative degree and their differentials: it suffices to take the
differentials of functions of degree $-1$. This is a coisotropic
ideal, and its vanishing locus is a coisotropic dg superscheme
$\M_+\subset\M$. The dg superalgebra of functions $\CO/\CO_-$ of the dg
superscheme $\M_+$ is concentrated in nonnegative degree.

The vanishing locus $\M_0=\M_+\cap\M_-$ of the ideal $\CO_++\CO_-$ is a
superscheme whose superalgebra of functions is
$\CO/(\CO_++\CO_-)$. This superalgebra may be identified with the
zeroth cohomology of $\CO/\CO_+$, and hence is the superalgebra of
functions on the critical locus $\pi^0(\M_-)$ of the dg superscheme
$\M_-$.

The following observation lies at the heart of the
Batalin--Fradkin--Vilkovisky (BFV) formalism.
\begin{lemma}
  The following conditions are equivalent:
  \begin{enumerate}[a)]
  \item the cohomology of $\CO$ vanishes in negative degree;
  \item the cohomology of $\CO/\CO_+$ vanishes in negative degree;
  \item the morphism $\CO/\CO_+\to\CO/(\CO_++\CO_-)$ is a
    quasi-isomorphism.
  \end{enumerate}
\end{lemma}

If $\CO$ satisfies these conditions, the dg supermanifold $\M$ is a
resolution of $\M_0$. Such resolutions are known to exist, and are
unique up to quasi-isomorphism of their dg Poisson superalgebras of
functions, under very general hypotheses \citep{HT,Stasheff}. The
deformation quantization of $\M$ is a dg superalgebra, which plays the
role of the derived deformation quantization of $\M_0$.

\subsection{The BRST formalism}

Let $\mu:M\to\g^\vee$ be a moment map. The BRST resolution $M(\mu)$ of the zero
locus $\mu^{-1}(0)\subset M$ is a dg supermanifold whose underlying graded
supermanifold is
\begin{equation*}
  M \times T^*\g[1] \cong M \times \g^\vee[-1] \times \g[1] .
\end{equation*}
Choose a basis $e_i$ of $\g$, with parity $\epsilon_i\in\{0,1\}$. Let
$C_{ij}^k$ be the structure coefficients of $\g$:
\begin{equation*}
  [e_i,e_j] = \sum_k C_{ij}^k e_k .
\end{equation*}
The graded supermanifold $\g[1]$ has coordinates $c^i$ of degree $1$,
the ghosts, with internal parity $\epsilon_i$, and total parity
$1-\epsilon_i$.  The graded supermanifold $\g^\vee[-1]$ has coordinates
$b_i$ of degree $-1$, the ghost momenta, also with internal parity
$\epsilon_i$ and total parity $1-\epsilon_i$. The symplectic form on
$T^*\g[1]$ equals
\begin{equation*}
  \sum_i (-1)^{\epsilon_i} db_i \wedge dc^i ,
\end{equation*}
with Poisson brackets $\{b_i,c^j\}=\delta_i^j$ and
$\{b_i,b_j\}=\{c^i,c^j\}=0$. The graded superalgebra of functions of
$M(\mu)$ is
\begin{equation*}
  \CO(M) \otimes \Sym\g[1] \otimes \Sym\g^*[-1] \cong \CO(M)[b_i,c^i] .
\end{equation*}
The differential is the Hamiltonian vector field associated to the
function
\begin{equation*}
  S = \sum_i \mu_i c^i + \frac{1}{2} \sum_{ijk} (-1)^{\epsilon_j} C_{ij}^k b_k c^j c^i .
\end{equation*}
The associated differential $\delta=\{S,-\}$ is
\begin{equation*}
  \delta = \sum_i (-1)^{\epsilon_i} \mu_i \frac{\partial}{\partial b_i} + \sum_i c^i \, \biggl(
  \{\mu_i,-\} + \sum_{jk} (-1)^{\epsilon_j+1} C_{ij}^k b_k \frac{\partial}{\partial b_j} \biggr)
  + \frac{1}{2} \sum_{ijk} (-1)^{\epsilon_j} C_{ij}^k c^j c^i \frac{\partial}{\partial c^k} .
\end{equation*}
The dg superalgebra of functions on the dg coisotropic supermanifold
$M(\mu)_-$ is the Koszul complex
$\CO(M\times\g^\vee[-1])\cong\CO(M)[b_i]$, with differential
\begin{equation*}
  \delta_0 = \sum_i (-1)^{\epsilon_i} \mu_i \frac{\partial}{\partial b_i} .
\end{equation*}
The complex as a whole may be identified with the Chevalley--Eilenberg
complex
\begin{equation*}
  C^*(\g,\CO(M\times\g^\vee[-1]))
\end{equation*}
of the dg $\g$-supermodule $\CO(M\times\g^\vee[-1])$.

The BRST dg symplectic manifold $M(\mu)$ associated to a moment map is
not a resolution of the zero-locus $\mu^{-1}(0)$ unless the zero-locus
is a complete intersection. An example where it is a resolution is the
massless relativistic particle, with $M=T^*U$,
$\Omega=dp_a \wedge e^a - p_a \omega^a{}_b \wedge e^b + B$ and
$\mu=\half\eta_{ab}p^ap_b:M\to\R$. In this case, the differential is the
Hamiltonian vector field
\begin{equation*}
  \delta = - \half\eta_{ab}p^ap_b \frac{\partial}{\partial b} + c\,\eta^{ab} f^\mu_a \, p_b
  \Bigl( \frac{\partial}{\partial x^\mu} - f^\nu_c B_{\mu\nu} \frac{\partial}{\partial p_c} \Bigr)
\end{equation*}
associated to the function $S = \half\eta_{ab}p^ap_b \, c$.

We now specialize to the case of $N=1$ supersymmetric quantum
mechanics. Denote the ghosts for the basis $Q$ and $H$ of the
associated Lie superalgebra $\h$ by $\gamma$ and $c$ respectively, with
conjugate momenta $\beta$ and $b$. The function $S$ equals
\begin{equation*}
  S = q \, \gamma + h \, c - b \, \gamma^2 ,
\end{equation*}
and the associated differential $\delta=\{S,-\}$ is
\begin{equation*}
  \delta = - q \frac{\partial}{\partial\beta} + h \frac{\partial}{\partial b} + \gamma \, \biggl( Q + 2b
  \frac{\partial}{\partial\beta} \biggr) + c \, Q^2 - \gamma^2 \frac{\partial}{\partial c} .
\end{equation*}

For the $N=1$ spinning particle, the zero locus $\mu^{-1}(0)$ is the
superlightcone
\begin{equation*}
  \{q=0\} \cap \{h=0\} .
\end{equation*}
This locus is not a complete intersection, because of the relation
$q\,\Omega=0$, where $\Omega=\theta^1\ldots\theta^D$. This is reflected in the dg superalgebra
of functions on the dg supermanifold $M(\mu)$ having cohomology classes
in all negative degrees \citep{curved}. We recall the proof in
Section~2.

The main result of this paper is that even though $M(\mu)$ is not the
BFV resolution of the superlightcone, its deformation quantization has
vanishing cohomology in negative degrees. It is tempting to say that
the superlightcone, while not a classical complete intersection, is
nevertheless a quantum complete intersection.

\subsection{The Moyal bracket}

Let $M$ be a symplectic supermanifold, with associated Poisson tensor
\begin{equation*}
  P = \frac12\sum_{\mu\nu} P^{\mu\nu}\partial_\mu\wedge\partial_\nu
\end{equation*}
so that
\begin{equation*}
  \sum_{\mu\nu} (-1)^{|f||x^\nu|} P^{\mu\nu} \partial_\mu f \partial_\nu g = \{f,g\} .
\end{equation*}
Choose a torsion-free symplectic connection $\nabla$ on $M$. The Fedosov
quantization of the algebra $\CO(M)\[\hbar\]$ associated to $\nabla$ is an
associative product
\begin{equation*}
  f\circ g = \sum_{k=0}^\infty \hbar^k C_k(f,g) ,
\end{equation*}
where $C_0(f,g)=fg$, $C_1(f,g)=\half\{f,g\}$, and
$C_k(g,f) = (-1)^{|f||g|+k} C_k(f,g)$ is a bidifferential operator of
order $k$ in each variable. (To reproduce the Heisenberg commutation
relation, the expansion parameter $\hbar=h/2\pi i$ is understood to vary
from the usual definition $\hbar=h/2\pi$ by a factor of $-i$.)

In this article, we focus on the associated graded Lie bracket
\begin{equation*}
  [f,g] = f\circ g - (-1)^{|f||g|} g\circ f = 2 \sum_{\ell=0}^\infty \hbar^{2\ell+1}
  C_{2\ell+1}(f,g) .
\end{equation*}
This bracket was introduced by Moyal~\citep{Moyal}, who also found
explicit formulas for $f\circ g$ and $[f,g]$ to all orders when $M$ is a
symplectic vector space and $P^{\mu\nu}$ is constant:
\begin{align*}
  C_k(f,g) &= \frac{1}{2^kk!} \sum_{\substack{\mu_1\ldots\mu_k\\\nu_1\ldots\nu_k}}
  (-1)^{\sum_{i<j}|x^{\nu_i}||x^{\nu_j}|+|f|(|x^{\nu_1}|+\cdots+|x^{\nu_k}|)} \\[-10pt]
           & \qquad\qquad\qquad\qquad\qquad P^{\mu_1\nu_1} \ldots P^{\mu_k\nu_k} \nabla_{\mu_1} \ldots \nabla_{\mu_k} f \,
             \nabla_{\nu_1}\ldots \nabla_{\nu_k}g .
\end{align*}

We will use the Fedosov product on $M$, realized by extending the
covariant complete Weyl calculus for pseudodifferential operators
(Pflaum \citep{Pflaum}) to the supermanifold $M$. (This is analogous
to the construction of \citep{AS}, where the complete calculus with
normal ordering is used.) In combination with the Moyal product in the
coordinates $\{\beta,b,\gamma,c\}$, this gives the Fedosov product on the
symplectic graded supermanifold $M(\mu)$. The function
\begin{equation*}
  S_\hbar = q \, \gamma + \tfrac{1}{2\hbar} [q,q] \, c - b \, \gamma^2 \in \CO(M(\mu))\[\hbar\] =
  \CO(M)[\beta,\gamma,b,c]\[\hbar\]
\end{equation*}
satisfies the Maurer--Cartan equation
\begin{equation*}
  \half [S_\hbar,S_\hbar] = S_\hbar\circ S_\hbar=0 ,
\end{equation*}
and we may consider the differential
\begin{equation*}
  \delta_\hbar = \tfrac{1}{\hbar} [S_\hbar,-] = \delta_0 + \hbar^2 \delta_1 + \cdots
\end{equation*}
it induces on $\CO(M(\mu))\(\hbar\)$. We call the cohomology of this complex
the Moyal cohomology of the BRST theory. It may be calculated by a
spectral sequence whose $E_1$ page is the cohomology of $\delta_0$, or the
cohomology of the BRST formalism. We show that for the $N=1$ spinning
particle, the negative degree cohomology classes in the cohomology of
$\delta_0$ do not survive to the $E_2$ page of the spectral sequence. We
discuss the case of constant frame $e^a_\mu$ and vanishing magnetic
potential $A_\mu$ as a warm-up in Section~3, and the general case in
Section~4.

Barnich and Grigoriev \citep{BG} relate the cohomology of the dg
manifold $M(\mu)$ to the cohomology of local observables in the
classical BV formalism for the AKSZ model associated to $M(\mu)$. This
raises the question of the role of the Moyal cohomology of $M(\mu)$ in
the BV formalism. For the relativistic particle, this question is
answer by by Grady, Li and Li \citep[Theorem 3.22]{GLL}, who relate
the Moyal cohomology of a symplectic manifold to the cohomology of
local observables in the quantum BV formalism for the associated AKSZ
model. The generalization of their work to symplectic dg
supermanifolds would illuminate the role that our result play in the
quantum BV formalism for the $N=1$ spinning particle.

In this article, we have focused on the cohomology of the complex of
operators of the quantized BFV theory. Boffo et al.\ \citep{BGHS} have
considered a related problem: they show that the complex of states of
the theory has vanishing cohomology in negative degree. This
calculation was the motivation for the present work.

\section{The Poisson cohomology of the {\boldmath$N=1$} spinning
  particle}

In this section, we review the calculation of the cohomology of
$\CO(M(\mu))$ for the $N=1$ spinning particle.  In these calculations,
we formally invert the variable $\gamma$: all negative powers of
$\gamma$ disappear in the course of the calculation, but working with
negative powers of $\gamma$ at intermediate steps in the calculation
simplifies the construction of cocycles.

The superscheme $M(\mu)_0\subset M(\mu)$ may be identified with
$\mu^{-1}(0)$, and the dg superalgebra of functions on the dg
superscheme $M(\mu)_+\subset M(\mu)$ may be identified with the
Chevalley--Eilenberg complex $C^*(\h,\CO(\mu^{-1}(0)))$. We have a
short exact sequence
\begin{equation*}
  \begin{tikzcd}[column sep=1.5em]
    0 \rar & \CO(M(\mu))_- \rar & \CO(M(\mu)) \rar &
    C^*(\h,\CO(\mu^{-1}(0))) \rar & 0 ,
  \end{tikzcd}
\end{equation*}
and the cohomology of $\CO(M(\mu))$ has two sources:
\begin{equation*}
  H^k(\CO(M(\mu))) \cong
  \begin{cases}
    H^k(\h,\CO(\mu^{-1}(0))) , & k\ge0 , \\
    H^k(\CO(M(\mu))_-) , & k<0 .
  \end{cases}
\end{equation*}
We focus on the cohomology in negative degree.

Let $\Omega=\theta^1\ldots\theta^D$. The divided powers of $\beta$ are
$\beta^{[k]}=\beta^k/k!$ for $k\ge0$, and $\beta^{[k]}=0$ for $k<0$.

Given a function $f\in\CO(U\times\g[1])$ of the
coordinates $\{x^\mu\}\cup\{\gamma,c\}$, introduce the functions
\begin{align*}
  \xi_k(f) &= - \gamma^{-1} c f \Omega \beta^{[1-k]} , &
  \eta_k(f) &= \gamma^{-1} f \Omega \beta^{[-k]} - \gamma^{-1} c \{\{q,f\},\Omega\} \beta^{[1-k]} ,
  \\
  X_k(f) &= \{S,\xi_k(f)\} , & Y_k(f) &= \{S,\eta_k(f)\} .
\end{align*}

By inspection, $q\Omega=\{q,f\}\Omega=\{h,\Omega\}=0$ vanish, and
\begin{equation*}
  \{q,\Omega\} = \sum_{a=1}^D (-1)^a p_a \, \theta^1 \ldots \widehat{\theta}{}^a \ldots \theta^D .
\end{equation*}
Since $q\Omega=0$, we see that
\begin{align*}
  \{h,f\}\Omega &= \{q,\{q,f\}\} \Omega = \{\{q,f\},q\} \Omega = \{\{q,f\},q \Omega\}
             + q \{ \{q,f\} , \Omega \} \\
           &= q \{ \{q,f\} , \Omega \} .
\end{align*}
Using these formulas, we calculate that
\begin{align*}
  X_k(f) &= \bigl( 2cb\beta^{[-k]} + \gamma \beta^{[1-k]} \bigr) f \Omega
           + c f \{q,\Omega\} \beta^{[1-k]} \intertext{and}
  Y_k(f) &= 2 b f \Omega \beta^{[-k-1]} + f \{q,\Omega\} \beta^{[-k]}
           + \bigl( 2cb\beta^{[-k]} + \gamma \beta^{[1-k]} \bigr) \{\{q,f\},\Omega\} \\
         &\qquad + c \{q,\{\{q,f\},\Omega\}\} \beta^{[1-k]} .
\end{align*}

\begin{theorem}
  \begin{equation*}
    H^k(\CO(M(\mu))_-) \cong
    \begin{cases}
      \{ X_k(f) + Y_k(g) \mid f,g\in\CO(U) \} , & k<0 , \\
      \{ X_0(f) \mid f\in\CO(U) \} , & k=0 , \\
      0 , & k>0 .
    \end{cases}
  \end{equation*}
\end{theorem}
\begin{proof}
  Filter the graded superspaces $\CO(M(\mu))_-$, placing the ghosts
  $\gamma$ and $c$ in filtration degree $2$ and $3$ respectively. The
  leading order of the differential on $\CO(M(\mu))_-$ equals
  \begin{equation*}
    \delta_{(0)} = - q \frac{\partial}{\partial\beta} + h \frac{\partial}{\partial b}
  \end{equation*}
  The cohomology of the differential $\delta_{(0)}$, the $E_1$ page of the
  spectral sequence for the filtration on $\CO(M(\mu))_-$, is easily
  calculated. Multiplication by $q$ gives a differential which, after
  conjugation by the Hodge star in the variables $\{\theta^a\}$, may be
  identified with the Koszul complex for the constraints
  $\{p_a=0\}$. The cohomology of this differential is spanned by the
  cocycle $\Omega$, multiplied by functions of the variables
  $\{x^\mu,\gamma,c\}$. We conclude that the cohomology of
  $\delta_{(0)}$ is the span of the following cocycles, where
  $f\in\CO(U)[\gamma,c]$ and $k<0$:
  \begin{align*}
    A_k(f) &= f \Omega \beta^{[-k]} , &
    B_k(f) &= 2b f \Omega \beta^{[-k-1]} + f \{q,\Omega\} \beta^{[-k]} .
  \end{align*}

  The differential $\delta_{(1)}$ of the $E_1$ page of the spectral
  sequence is
  \begin{equation*}
    \delta_{(1)} = - \gamma^2 \frac{\partial}{\partial c} .
  \end{equation*}
  Its cohomology is straightforward to calculate, and shows that $E_2$
  is spanned by the cocycles
  \begin{align*}
    \{ A_k(f_0+\gamma f_1) + B_k(g_0 + \gamma g_1) \mid f_0,f_1,g_0,g_1\in\CO(U) \}
    , \quad k<0 .
  \end{align*}
  
  The differential $\delta_{(2)}$ of the $E_2$ page of the spectral
  sequence is
  \begin{equation*}
    \delta_{(2)} = \gamma Q + 2b\gamma \frac{\partial}{\partial\beta} .
  \end{equation*}
  On the cocycles $A_k(f_0+\gamma f_1)$ and $B_k(g_0+\gamma g_1)$, the
  differential $\delta_{(2)}$ acts as
  \begin{align*}
    \delta_{(2)} A_k(f_0+\gamma f_1) &= B_k(\gamma f_0) , &
    \delta_{(2)} B_k(g_0+\gamma g_1) &= 0 .
  \end{align*}
  Its cohomology is spanned by $A_k(\gamma\,f)$ and $B_k(g)$, $f,g\in\CO(U)$.

  The $E_\infty$ page of the spectral sequence is isomorphic to the
  $E_3$ page. The classes $A_k(\gamma f)$ and $B_k(g)$ lift to
  $X_{k+1}(f)$ and $Y_k(g)$ respectively.
\end{proof}

Since the Lie superalgebra is free, its cohomology is easily
calculated: there is a quasi-isomorphism between the
Chevalley--Eilenberg complex $C^*(\h,\CO(\mu^{-1}(0)))$ and the complex
\begin{equation*}
  \begin{tikzcd}[column sep=2em]
    \CO(\mu^{-1}(0)) \rar{\gamma\,Q} & \gamma\,\CO(\mu^{-1}(0)) .
  \end{tikzcd}
\end{equation*}
In particular, $H^k(\CO(M(\mu)))$ vanishes for $k>1$.

\begin{corollary}
  \begin{equation*}
    H^k(\CO(M(\mu))) \cong
    \begin{cases}
      \{ X_k(f) + Y_k(g) \mid f,g\in\CO(U) \} , & k<0 , \\
      H^0(\h,\CO(\mu^{-1}(0))) \oplus \{ X_0(f) \mid f\in\CO(U) \} , & k=0 , \\
      H^1(\h,\CO(\mu^{-1}(0))) , & k=1 , \\
      0 , & k>1 .
    \end{cases}
  \end{equation*}
\end{corollary}

\section{The Moyal cohomology of the {\boldmath$N=1$} spinning
  particle: flat case}

In this section, we calculate the Moyal cohomology of $M(\mu)$ in a
Euclidean background, where the moving frame $e^a_\mu$ is the Kronecker
delta $\delta^a_\mu$ and the magnetic potential $A_\mu$ vanishes. In this case,
the formula for $S_\hbar$ simplifies, since the Moyal bracket $[q,q]$ of
$q$ with itself equals the Poisson bracket $\hbar\{q,q\}$. Thus
$S_\hbar$ equals $S$:
\begin{equation*}
  S_\hbar = S = \frac{1}{2} \delta^{ab} p_ap_b\,c + p_a\theta^a\,\gamma - \gamma^2 b .
\end{equation*}
Since this is a cubic function, there are only two terms in the Moyal
bracket $\delta_\hbar=[S_\hbar,-]$: the leading term, which is the Poisson bracket,
and a differential operator of degree $3$:
\begin{align*}
  \delta_\hbar &= \hbar \delta_0 + \hbar^3 \delta_1 \\
      &= \hbar \{ S , - \} + \frac{\hbar^3}{4} \biggl( \frac{1}{2} \delta^{ab}
        \frac{\partial^3}{\partial x^a\partial x^b\partial b} - \delta^{ab} \frac{\partial^3}{\partial x^a \partial \theta^b \partial\beta} -
        \frac{\partial^3}{\partial\beta^2\partial c} \biggr) .
\end{align*}
By the equation $\delta_\hbar^2=0$, we see that
$\delta_0\delta_1+\delta_1\delta_0=0$ and $\delta_1^2=0$. Thus, in the flat case, the
differential $\delta_\hbar$ is the total differential of a double complex. We
now study the spectral sequence of this double complex: this is
associated to the filtration where $\hbar$ is given filtration degree $1$.

We have already calculated the cohomology of $\delta_0$ in the previous
section. In negative degrees, it is spanned by
$X_k(f)=\delta_0\xi_k(f)$ and $Y_k(f)=\delta_0\eta_k(f)$, $k>0$, and in degrees
$0$ and $1$, it is spanned by $X_0(f)$ and
$H^*(\h,\CO(\mu^{-1}(0)))$. The cocycles spanning the summand
$H^*(\h,\CO(\mu^{-1}))$ of $H^*(\CO(M(\mu)))$ have no dependency on
$\beta$ or $b$, and hence $\delta_1$ vanishes on $H^*(\h,\CO(\mu^{-1}(0)))$.

The following calculation is the main result of this article, and
completes the proof that the Moyal cohomology vanishes outside degrees
$0$ and $1$ in this special case.
\begin{lemma}
  \begin{equation*}
    \delta_1\xi_k(f) = \tfrac{1}{4} \eta_{k-1}(f)
  \end{equation*}
\end{lemma}
\begin{proof}
  Since $\partial\xi_k(f)/\partial b=0$, we see that
  \begin{align*}
    \frac{\partial^3\xi_k(f)}{\partial x^a\partial x^b\partial b} &= 0 \intertext{We have}
    \frac{\partial^3\xi_k(f)}{\partial x^a\partial\theta^b\partial\beta} &= \gamma^{-1} c \biggl(
    \delta^{ab} \frac{\partial f}{\partial x^a} \frac{\partial\Omega}{\partial\theta^b} \biggr) \beta^{[k-1]}
    = \gamma^{-1} c  \{ \{ q,f \} , \Omega \} \beta^{[k-1]} \intertext{and}
    \frac{\partial^3\xi_k(f)}{\partial\beta^2\partial c} &= - \gamma^{-1} f \Omega \beta^{[k-2]} .
    \qedhere
  \end{align*}
\end{proof}

\begin{corollary}
  \begin{equation*}
    \delta_1X_k(f) = - \tfrac{1}{4} Y_{k-1}(f)
  \end{equation*}
\end{corollary}
\begin{proof}
  This follows from the equation $\delta_0\delta_1+\delta_1\delta_0=0$.
\end{proof}

We conclude that the $E_2$ page of the spectral sequence, and hence
the Moyal cohomology itself, is concentrated in degrees $0$ and $1$.

\section{The Moyal cohomology of the {\boldmath$N=1$} spinning
  particle: general case}

In this section, we extend the calculation of the last section to the
general case, in the presence of curvature and a magnetic field. In
order to do this, we will identify the Fedosov deformation
quantization with the complete symbol calculus for differential
operators acting on spinors over the Riemannian manifold $U$.

Let $U$ be a manifold with affine connection $\nabla^U$, and let $E$ be a
vector bundle over $U$ with covariant derivative $\nabla^E$. Let
$\CD$ be a differential operator on $E$. Let $\pi^*\End(E)$ be the
pullback to the cotangent bundle $\pi:T^*U\to U$ of the bundle
$\End(E)$ over $U$. The complete symbol
\begin{equation*}
  \sigma(\CD) \in \Gamma(T^*U,\pi^*\End(E))
\end{equation*}
of $\CD$ is a section of $\pi^*\End(E)$ polynomial in the
fibers. Restricted to the fiber $T^*_xU$, $x\in U$, $\sigma(\CD)$ is the
complete Weyl symbol of the germ of the differential operator $\CD$ in
the normal coordinate system around $x$ defined by the connection
$\nabla^U$, with respect to the radial gauge on the vector bundle $E$. The
covariant complete Weyl symbol (Pflaum \citep{Pflaum}) is
characterized by the formula
\begin{equation*}
  (\CD f)(x) = \hbar^{-D} \int_{T_xU\times T^*_xU} e^{ip\cdot(x-y)/\hbar} \,
  \sigma(\CD)(\half(x+y),p) \, f(y) \, dy\,dp .
\end{equation*}
The Fedosov quantization of $\Gamma(T^*M,\pi^*\End(E))$ is given by
\begin{equation*}
  \sigma(\CD_1) \circ \sigma(\CD_2) = \sigma(\CD_1\CD_2) .
\end{equation*}

The Fedosov quantization of the supermanifold $M$ will be realized as
using the algebra of differential operators on $U$ acting on a spinor
bundle, with coefficients in $\C\[\hbar^{1/2}\]$. Let $C(D-1,1)$ be the
Clifford algebra of the inner-product space $\C^{D-1,1}$, with inner
product $\eta^{ab}/2$. Denote the generators of $C(D-1,1)$ by
$\{c^a\}$.

If $D=\dim(U)$ is even, the algebra $C(D-1,1)$ may be identified with
$\End(\S)$, where $\S$ is the spinor bundle, with fiber
$\C^{2^{D/2}-1|2^{D/2}-1}$ and connection
\begin{equation*}
  \nabla^\S_\mu = \partial_\mu - \tfrac{1}{2} \omega_{\mu ab} c^ac^b .
\end{equation*}
The odd coordinates $\theta^a$ of $M$ act on $\S$ by $\hbar^{1/2} c^a$.

If $D$ is odd, we consider the Clifford algebra $C(D,1)$ with
generators $\{c^0,c^a\}$; the Clifford algebra $C(D-1,1)$ is the
centralizer of $c^0$.  In this case, the spinor bundle has fiber
$\C^{2^{(D-1)/2}|2^{(D-1)/2}}$, and $\nabla^\S c^0=0$. In both the even and
the odd cases, we have $\nabla^\S\Omega=0$ and
\begin{equation*}
  \Omega\circ\Omega = (-1)^{\binom{D}{2}+1} = (-1)^{\bigl\lfloor \frac{D}{2} \bigr\rfloor+1} .
\end{equation*}
The additional factor of $-1$ is due to the Lorentzian signature of
the inner product $\eta^{ab}$.

The complete symbol of the Dirac operator $\hbar\D$ is
$\hbar^{-1/2}q$. It is proved by Pflaum \citep[Proposition 2.4]{Pflaum}
that the covariant complete Weyl symbol of the Laplacian
$-\eta^{ab}\nabla^2_{ab}$ associated to a vector bundle $E$ with covariant
derivative $\nabla$ is $\delta^{ab}p_ap_b$. Taking for $E$ a complex line bundle
with curvature the two-form $B$, and applying the Lichernowicz
formula, we see that the complete Weyl symbol of the operator
$\hbar^2\D^2=\half[\hbar\D,\hbar\D]$ equals
\begin{equation*}
  \half \<q,q\> = \half \{q,q\} + \tfrac{1}{4} \hbar^2 R ,
\end{equation*}
where $R$ is the scalar curvature of the Riemannian metric on $U$. (In
fact, the formula for the complete symbol of $\<q,q\>$, and certainly
the coefficient multiplying the scalar curvature, will not be used
further in our calculations.)

The following lemma follows from the explicit formula for the action
of the Clifford algebra $C(D-1,1)$ on the spinor bundle $\S$.
\begin{lemma}
  $[q, f\circ \Omega] = \half [[ q , f ] , \Omega ]$
\end{lemma}
\begin{proof}
  An explicit calculation shows that
  $q\circ\Omega=-(-1)^D\Omega\circ q$. Since $[f,\Omega]=0$, we see that
  \begin{align*}
    \half [[ q , f ] , \Omega ] &= \half (q\circ f-f\circ q)\circ\Omega -
    \half (-1)^D \Omega\circ(q\circ f-(-1)^Df\circ q) \\
    &= (q\circ f-f\circ q)\circ\Omega = q\circ f\circ\Omega + (-1)^Df\circ\Omega\circ q .
    \qedhere
  \end{align*}
\end{proof}

The function
\begin{equation*}
  S_\hbar = - q\gamma + \half \<q,q\>c - b\gamma^2 = S + \tfrac{\hbar^2}{4} R c
\end{equation*}
satisfies the Maurer--Cartan equation
\begin{equation*}
  S_\hbar \circ S_\hbar = 0 .
\end{equation*}
The Moyal cohomology is the cohomology with respect to the
differential $[S_\hbar,-]$.

We now imitate the calculations of Sections 1 and 2. We choose for
$\xi_k(f)$ the same expression as in the classical limit:
\begin{equation*}
  \fbox{\vphantom{$\Bigm|$}\ $\xi_k(f) = - \gamma^{-1}\beta^{[k]} c f \Omega$\ }
\end{equation*}
We have
\begin{align*}
  \<\gamma^2b,\xi_k(f)\>
  &= \<\gamma^2b,\gamma^{-1}\beta^{[k]}c\>\circ f\circ\Omega \\
  &= \bigl( \gamma \beta^{[k]} - 2 \beta^{[k-1]} bc \bigr)\circ f\circ \Omega + \tfrac{1}{4}
    \hbar^2 \gamma^{-1} \beta^{[k-2]} \circ f\circ \Omega .
\intertext{To calculate $\<S_\hbar,\xi_k(f)\>$, we proceed as follows:}
  \gamma q \circ \xi_k(f)
  &= \gamma \circ (\gamma^{-1}\beta^{[k]}c) \circ q \circ f\circ \Omega \\
  &= \bigl( \beta^{[k]} c - \half \hbar \gamma^{-1} \beta^{[k-1]} c \Bigr) \circ q \circ
    f\circ \Omega \\
  \xi_k(f) \circ \gamma q
  &= \bigl( - \gamma^{-1}\beta^{[k]}c \bigr) \circ \gamma \circ f \circ \Omega\circ q \\
  &= - \bigl( \beta^{[k]} c + \half \hbar \gamma^{-1} \beta^{[k-1]} c \bigr) \circ
    f\circ \Omega\circ q .
\intertext{It follows that}
  \<\gamma q,\xi_k(f)\>
  &= \hbar^{-1} \beta^{[k]} c \circ \bigl( q \circ f\circ \Omega - (-1)^d f\circ \Omega\circ q \bigr) \\
  &\quad - \half \gamma^{-1} \beta^{[k-1]} c \circ \bigl( q \circ f\circ \Omega + (-1)^d f\circ
    \Omega\circ q \bigr) \\
  &= \beta^{[k]} c \circ \<q,f\circ\Omega\>
    - \tfrac{1}{4} \hbar^2 \gamma^{-1} \beta^{[k-1]} c \circ \<\<q,f\>,\Omega\> .
\intertext{We see that}
  \< S_\hbar , \xi_k(f) \>
  &= \Bigl( \bigl( \gamma \beta^{[k]} - 2 \beta^{[k-1]} bc \bigr)\circ f\circ \Omega + \beta^{[k]}
    c \circ \<q,f\circ\Omega\> \Bigr) \\
  &\quad - \tfrac{1}{4} \hbar^2 \Bigl( \gamma^{-1} \beta^{[k-2]} \circ f\circ \Omega + \gamma^{-1}
    \beta^{[k-1]} c \circ \<\<q,f\>,\Omega\> \Bigr) .
\end{align*}

This calculation allows us to extend the proof of the lemma of
Section~2 to general backgrounds. The only change in the statement is
that we define $\eta_k(f)$ using the Moyal bracket and not the Poisson
bracket: it equals the coefficient of $\hbar^2$ in the Moyal bracket
$\<S_\hbar,\xi_{k+1}(f)\>$:
\begin{equation*}
  \fbox{\vphantom{$\Bigm|$}\ $\eta_k(f) = \gamma^{-1} \beta^{[k]} c \circ
    \<\<q,f\>,\Omega\> + \gamma^{-1} \beta^{[k-1]} \circ f\circ \Omega$\ }
\end{equation*}

\subsection*{Acknowledgements}

{\small
  The research contained in this article was performed while the
  author held a Visiting Professorship in the Department of
  Mathematics and Center for Geometry and Physics at Uppsala
  University, funded by the Knut and Alice Wallenberg Foundation.

The author gratefully acknowledge the support of the Mittag-Leffler
Institute: this article was conceived during the program
``Cohomological Aspects of Quantum Field Theory.'' We are also
grateful to the other participants in this program, and in particular
Eugenia Boffo, Martin Cederwall and Ivo Sachs, for fruitful
conversations on this work.

Markus Pflaum informed us of the covariant complete Weyl calculus used
in Section~3.}



\end{document}